\pdfoutput=1

\documentclass[11pt]{article} 
\usepackage[margin=1in]{geometry}
\usepackage[ruled,noend,noline]{algorithm2e}
\usepackage{amsfonts}
\usepackage{amsmath}
\usepackage{amsthm}
\usepackage{url}
\usepackage{wrapfig}
\usepackage{subcaption}
\usepackage{subfig}

\newcommand{\R}{\mathbb{R}}

\newcommand{\eps}{\varepsilon}
\renewcommand{\epsilon}{\varepsilon}

\newcommand{\opt}{o}

\renewcommand{\paragraph}[1]    {\medbreak\noindent \textbf{#1}}

\newcommand{\nnd}{\mathcal{D}}
\newcommand{\dtype}[1]{{\sf\small #1}}
\newcommand{\libraryname}[1]{{\sc #1}}

\usepackage{pgfplots}
\usepackage{pgfplotstable}
\usepackage{booktabs}
\usetikzlibrary{patterns}
\usetikzlibrary{external}

\pgfplotsset
{
    tick label style={font=\footnotesize},
    legend style    ={font=\footnotesize},
    label style     ={font=\footnotesize},
    tickwidth       =2pt,
    ylabel near ticks,
    xlabel near ticks,
    legend style     ={draw=none},
    legend cell align=right,
    legend plot pos  =right,
    /pgfplots/narrow area legend/.style=
    {
        legend image code/.code={ \draw[#1] (0cm,-0.06cm) rectangle (0.6cm,0.06cm); }
    },
    /pgfplots/short line legend/.style=
    {
        legend image code/.code={ \draw[#1] (0cm,0cm) rectangle (0.6cm,0cm); }
    },
}
\pgfplotscreateplotcyclelist{filled-color}
{
    every mark/.append style={solid,fill=blue!80!white},    mark=*          \\%
    every mark/.append style={solid,fill=red!80!white},     mark=square*    \\%
    every mark/.append style={solid,fill=orange!80!white},  mark=pentagon*  \\%
    every mark/.append style={solid,fill=green},            mark=diamond*   \\%
    every mark/.append style={solid,fill=black},            mark=triangle*  \\%
}

\newcommand{\addaveraged}[2][]
           {\addplot+[#1,error bars/.cd, y dir=both,y explicit] table[x index=0, y index=1, y error index=2] {#2};}

\pgfplotsset{
    geommark/.style={black,mark options={solid,fill=red!80!white},mark=square*},
    nongeommark/.style={black,mark options={solid,fill=blue!80!white},mark=*},
    dionmark/.style={black,mark options={solid,fill=orange!80!white},mark=pentagon*},
    greentrianglemark/.style={black,mark options={solid,fill=green!80!white},mark=triangle*},
    browndiamondmark/.style={black,mark options={solid,fill=brown!80!white},mark=diamond*}
}

\newtheorem{theorem}{Theorem}
\newtheorem{lemma}[theorem]{Lemma}
\newtheorem{proposition}[theorem]{Proposition}

\title{Geometry Helps to Compare Persistence Diagrams}
\author{Michael Kerber
    \footnote{Graz University of Technology, Graz, Austria \texttt{kerber@tugraz.at}}
\and
Dmitriy Morozov
\footnote{Lawrence Berkeley National Laboratory,
Berkeley, CA,  USA \texttt{dmitriy@mrzv.org}}
\and
Arnur Nigmetov
\footnote{Graz University of Technology, Graz, Austria \texttt{nigmetov@tugraz.at}}
}

\begin{document}

\date{}

\maketitle

\begin{abstract}
Exploiting geometric structure to improve the asymptotic complexity of discrete
assignment problems is a well-studied subject.  In contrast, the practical
advantages of using geometry for such problems have not been explored.
We implement geometric variants of the Hopcroft--Karp algorithm for bottleneck
matching (based on previous work by Efrat el al.) and of the auction algorithm
by Bertsekas for Wasserstein distance computation.
Both implementations use k-d trees to replace a linear scan with
a geometric proximity query.
Our interest in this problem stems from the desire to compute
distances between persistence diagrams, a problem that comes up frequently
in topological data analysis.
We show that our geometric matching algorithms lead to a substantial
performance gain, both in running time and in memory consumption,
over their purely combinatorial counterparts.
Moreover, our implementation significantly outperforms the only other implementation
available for comparing persistence diagrams.
\end{abstract}

\section{Introduction}

The \emph{assignment problem} is among the most famous problems in
combinatorial optimization. Given a weighted bipartite
graph $G$ with $(n+n)$ vertices, it asks for a perfect matching with minimal cost.
A common cost function is the minimum of the sum of the $q$-th powers of weights
of the matching edges, for some $q \geq 1 $. We call the solution
in this case the \emph{$q$-Wasserstein matching}
and its cost
the \emph{$q$-Wasserstein distance}.
As $q$ tends to infinity, the Wasserstein distance approaches the
\emph{bottleneck distance}, by definition the minimum of the maximum edge weight over all perfect matchings.
See~\cite{bdm-assignment} for a contemporary discussion of the topic
with links to applications.

We consider the geometric version of the assignment problem, where the vertices
of $G$ are points in a metric space $(X,d)$, and edge weights are determined
by the distance function $d$.
The metric
structure leads to asymptotically improved algorithms that take advantage of
data structures for near-neighbor search. This line of research
dates back to Efrat et al.~\cite{eik-geometry} for the bottleneck distance
and Vaidya~\cite{vaidya-geometry} for the $1$-Wasserstein case. Rich literature has
developed since then, mainly focusing on approximation algorithms for Euclidean metrics
in low and high dimensions; see~\cite{as-approximation} for a recent summary.
On the other hand, there has been no rigorous study of whether
geometry also helps \emph{in practice}.
Our paper is devoted to this question.

We restrict attention to one scenario that motivates our study of the assignment problem.
In the field of \emph{topological data analysis}, the homological information
of a data set is often summarized in a \emph{persistence diagram}.
Such diagrams, themselves point sets in $\R^2$, capture
connectivity of a data set, and, specifically, how the connectivity changes
across various scales~\cite{elz-topological}.
Persistence diagrams are stable: small changes in the data
cause only small changes in the diagram~\cite{ceh-stability,wasserstein-stability}.
Accordingly, the distances between persistence diagrams
have received a lot of attention in applications (e.g., \cite{arc-classification,ggk-topology,gh-exploring}):
where persistence diagrams serve as topological proxies for the input data,
distances between the diagrams serve as proxy measures of the similarity between
data sets.
These distances, in turn, can be expressed as a Wasserstein or a bottleneck
distance between two planar point sets,
using $L_\infty$ as the metric in the plane
(see Section~\ref{sec:prelim} for the precise definition and the reduction).

\paragraph{Our contributions.}
Our contribution is two-fold. First, we provide an experimental study illuminating the advantages
of exploiting geometric structure in assignment problems: we compare mature implementations
of bottleneck and Wasserstein distance computations for the geometric and purely combinatorial versions
of the problem and demonstrate that exploiting the spatial structure improves
running time and space consumption for the matching problem.
Second, by focusing on the setup relevant in topological data analysis,
we provide the fastest implementation for computing distances between persistence diagrams,
significantly improving the implementation in the \libraryname{Dionysus} library~\cite{dionysus}.
The latter prototypical implementation is the only publicly available software for
the problem. Given the importance of this problem in applications,
our implementation is therefore addressing a real need in the community.
Our code is publicly available.
\footnote{\url{https://bitbucket.org/grey_narn/hera}}
This paper contains the following specific contributions:

\begin{itemize}

\item For bottleneck matchings, we follow the approach of Efrat et al.~\cite{eik-geometry}:
they augment the classical combinatorial algorithm of Hopcroft and Karp~\cite{hk-algorithm}
with a geometric data structure to speed up the search for vertices close to query points.
We do not implement their asymptotically optimal but complicated approach.
We instead use a k-d tree data structure~\cite{kd-tree}
to prune the search for matching vertices in remote areas (also proposed by the authors).
As expected, this strategy outperforms the combinatorial version that
linearly scans all vertices.
Several careful design choices are necessary
to obtain this improvement; see Section~\ref{sec:bottleneck}.

\item
For Wasserstein matchings, we implement a geometric
variant of the \emph{auction algorithm}, an approximation algorithm
by Bertsekas~\cite{bertsekas-distributed}.
We use \emph{weighted} k-d trees, again with the goal to
reduce the search range when looking for the best match of a vertex.
A data structure similar to ours appears in~\cite{as-wann}.
Our implementation outperforms
a version of the auction algorithm that does not exploit geometry,
which we implement for comparison, both in terms
of runtime and space consumption.
Both our implementations of the auction algorithm
dramatically outperform \libraryname{Dionysus}, albeit computing
approximations rather than the exact answers as the latter.
\libraryname{Dionysus} uses a variant of the Hungarian
algorithm~\cite{munkres-algorithms};
see Section~\ref{sec:wasserstein}.

\item
We extend our auction implementation to the case of
points with multiplicities, or masses. While this problem can be trivially
reduced to the previous one by replacing a multiple point
with a suitable number of simple copies,
it is more efficient to handle a point with multiplicity as one entity,
splitting it adaptively only when
fractions are matched to different points.
An extension of the auction algorithm to this case
has been decribed by Bertsekas and Casta{\~n}on~\cite{bc-auction}.
We refer to it as \emph{auction with integer masses}.
Our implementation exploits the geometry of the problem in a similar
way as the auction for simple points.
Handling masses imposes a certain overhead that slows
down the computation if the multiplicities are low. However,
our experiments show that the advantage of the auction with integer masses 
becomes apparent already when the average multiplicity is around $10$,
and the performance gap between the two variants of the auction 
increases when the average multiplicity increases;
see Section~\ref{sec:weighted_wasserstein}.

\end{itemize}

A conference version of this article appeared in ALENEX 2016~\cite{kmn-geometry}.
The major novelty of the present version is the discussion of the auction
with integer masses in Section~\ref{sec:weighted_wasserstein}. 
Moreover, we employed a different
variant in the (standard) auction algorithm, which improved the running time
of the geometric version by more than a magnitude.
Technical explanations and updated experimental evaluation compared to~\cite{kmn-geometry}
are discussed in Section~\ref{sec:wasserstein}.

\section{Background}
\label{sec:prelim}

\paragraph{Assignment problem.}
Given a weighted bipartite graph $G=(A\sqcup B,E,w)$, with $|A|=n=|B|$ and a weight function $w:E\rightarrow \R_+$,
a \emph{matching} is a subset $M\subseteq E$ such that every vertex of $A$ and of $B$ is
incident to at most one edge in $M$. These vertices are called \emph{matched}. A matching is \emph{perfect}
if every vertex is matched; equivalently, a perfect matching is a matching of cardinality $n$;
it can be expressed as a bijection $\eta:A\rightarrow B$.

For a perfect matching $M$, the \emph{bottleneck cost} is defined as $\max\{w(e)\mid e\in M\}$,
the maximal weight of its edges. The \emph{$q$-Wasserstein cost} is defined
as $(\sum_{e\in M} w(e)^q)^{1/q}$; for $q=1$, this is simply the sum of the edge weights.
A perfect matching is \emph{optimal}
if its cost is minimal among all perfect matchings of $G$.
In this case, the \emph{bottleneck} or \emph{$q$-Wasserstein cost} of $G$ is the cost
of an optimal matching. If a graph does not have a perfect matching, its
cost is infinite. For $q>1$, the $q$-Wasserstein cost can be reduced to the case $q=1$
with the following simple observation.

\begin{proposition}
\label{prop:q_to_1_reduction}
The $q$-Wasserstein cost of $G=(A\sqcup B,E,w)$ equals $q$-th root of
the $1$-Wasserstein cost of $G'=(A\sqcup B,E,w^q)$, where $w^q$ means
that all edge weights are raised to the $q$-th power.
\end{proposition}

We call a graph $G=(A\sqcup B,E,w)$ \emph{geometric}, if there exists a metric space $(X,d)$
and a map $\phi:A\sqcup B\rightarrow X$ such that for any edge $e=(a,b)\in E$,
$w(e)=d(\phi(a),\phi(b))$. In this case, we generally blur the distinction between
vertices and their embedding and just assume for simplicity that $A\sqcup B\subset X$.
The motivating example of this work is $X=\R^2$ and $d(x,y)=\|x-y\|_\infty$.

\paragraph{Persistent homology and diagrams.}
We are concerned with a particular type of assignment problems in this paper.
Specifically, we are interested in distances studied by the \emph{theory of
persistent homology}, distances that measure topological
differences between objects.
In a nutshell, persistent homology records connectivity of objects ---
connected components, tunnels, voids, and higher-dimensional ``holes''
--- across multiple scales.
\emph{Persistence diagrams} summarize this information as two-dimensional
point sets with multiplicities.
A point $(x,y)$ with multiplicity $m$ represents $m$ features that all appear
for the first time at scale $x$ and disappear at scale $y$.
Features appear before they disappear, so the points lie above the diagonal $x=y$.
The difference $y-x$ is called the \emph{persistence} of a feature.
In addition to the off-diagonal points, the persistence diagram
also contains each diagonal point $(x,x)$, counted with infinite multiplicity.
These additional points are needed for stability (discussed below)
and make the cardinality of every persistence diagram infinite,
even if the number of off-diagonal points is finite.

Given two persistence diagrams $X$ and $Y$, their \emph{bottleneck distance} is defined as
\[W_\infty(X,Y) = \inf_{\eta:X\rightarrow Y} \sup_{x\in X} \|x-\eta(x)\|_\infty,\]
where $\eta$ ranges over all bijections and $\|(x,y)\|_\infty=\max\{|x|,|y|\}$ is the usual $L_\infty$-norm.
Similarly, the \emph{$q$-Wasserstein distance}\footnote{
    Named after Leonid Va\v{s}erste\u{\i}n; see A. M. Vershik \cite{vershik2013long}
for the history of this notion.}
is defined as
\[W_q(X,Y) = \left[ \inf_{\eta:X\rightarrow Y} \sum_{x\in X} \|x-\eta(x)\|_\infty^q\right]^{1/q}.\]

Why are these distances interesting?
Because they are stable~\cite[Ch.\ VIII.3]{ceh-stability,wasserstein-stability,eh-computational}:
a small perturbation of the measured phenomenon, for example, a scalar function
on a manifold, creates only a small change in the persistence diagram --- both
distances reflect this.
The diagonal of a persistence diagram plays a crucial role in stability. Small
perturbations may create new topological features, but their persistence is
necessarily small, making it possible to match them to the points on the
diagonal. We refer the reader to the cited papers for an extensive
discussion.

\paragraph{Persistence distance as a matching problem.}

We assume from now on that persistence diagrams
consist of finitely many off-diagonal points with finite multiplicity
(and all the diagonal points with infinite multiplicity).
In this case, the task of computing $W_\ast(X,Y)$ can be reduced
to a bipartite graph matching problem; we follow the notation and
argument given in~\cite[Ch. VIII.4]{eh-computational}.
Let $X_0$, $Y_0$ denote the off-diagonal points of $X$ and $Y$, respectively.
If $u = (x,y)$ is an off-diagonal point, we denote its orthogonal projection on the diagonal $((x+y)/2,(x+y)/2)$ as $u'$,
which is the closest point to $u$ on the diagonal.
Let $X_0'$ denote the set of all projections of $X_0$, that is
$X_0'=\{ u' \mid u\in X_0\}$. With $Y_0'$ defined analogously as $\{ v' \mid v\in Y_0\}$, we define
$U=X_0\cup Y_0'$ and $V=Y_0\cup X_0'$; both have the same number of points.
We define the weighted complete bipartite graph,
$G=(U\sqcup V,U\times V,c)$,
whose weights are given by the function
\begin{figure}[t]
 \centering

 \begin{tikzpicture}[scale=0.6]
	 \tikzstyle{vertex_a}=[circle,minimum size=4pt,inner sep=0pt,fill=blue]
	 \tikzstyle{vertex_b}=[rectangle,minimum size=4pt,inner sep=0pt,fill=red]

	 \tikzstyle{selected edge} = [draw,line width=1pt,-,green]
	 \tikzstyle{diagonal edge} = [draw,line width=1pt,-,black]
	 \tikzstyle{unselected edge} = [draw,line width=0.5pt,-,gray!80]

	 \node[vertex_a] (a1) at (1,5) {};
	 \node[vertex_a] (a2) at (3,7) {};
	 \node[vertex_a] (b1p) at (4,4) {};
	 \node[vertex_a] (b2p) at (6, 6) {};

	 \node[vertex_b] (b1) at (1,7) {};
	 \node[vertex_b] (b2) at (3, 9) {};
	 \node[vertex_b] (a1p) at (3, 3) {};
	 \node[vertex_b] (a2p) at (5, 5) {};

	 \draw[selected edge] (a1) -- (b1);
	 \draw[selected edge] (a1) -- (b2);
	 \draw[selected edge] (a2) -- (b1);
	 \draw[selected edge] (a2) -- (b2);

	 \draw[selected edge] (a1) -- (a1p);
	 \draw[selected edge] (a2) -- (a2p);
	 \draw[selected edge] (b1) -- (b1p);
	 \draw[selected edge] (b2) -- (b2p);

	 \draw[unselected edge] (a1) -- (a2p) [dashed];
	 \draw[unselected edge] (a2) -- (a1p) [dashed];
	 \draw[unselected edge] (b1) -- (b2p) [dashed];
	 \draw[unselected edge] (b2) -- (b1p) [dashed];

     \draw[diagonal edge] (a1p) -- (b1p) [dotted];
     \draw[diagonal edge] (a2p) -- (b1p) [dotted];
     \draw[diagonal edge] (a2p) -- (b2p) [dotted];
     \draw[diagonal edge] (b2p) to [out=-90,in=0] (a1p) [dotted];

\end{tikzpicture}
\caption{An example of $G$ for two persistence diagrams with 2 off-diagonal points each. Skew edges are dashed gray, edges connecting diagonal points are dotted black.}
\label{fig:matching-graph}
\end{figure}
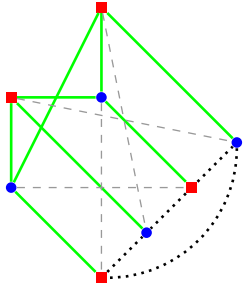

\begin{equation}
    \label{eq:c_q}
c(u,v)=\begin{cases} \|u-v\|_\infty & \text{if}~u\in X_0 ~\text{or}~ v\in Y_0\\ 0 & \text{otherwise}\end{cases}.
\end{equation}
Points from $U$ and $V$ are depicted as squares and circles, respectively, in Figure~\ref{fig:matching-graph} on the left; all the diagonal
points are connected by edges of weight 0 (plotted as dotted lines).
The following result is stated as the \emph{Reduction lemma} in~\cite[Ch. VIII.4]{eh-computational}:
\begin{lemma}
\label{lem:reduction_lemma}
~
\begin{itemize}
\item $W_\infty(X,Y)$ equals the bottleneck cost of $G$.
\item $W_q(X,Y)$ equals the $q$-Wasserstein cost of $G$.
This is equal to the $q$-th root of the $1$-Wasserstein cost of $G^q$,
which is the graph $G$ with cost function $c^q$, raising all edge costs
to the $q$-th power.
\end{itemize}
\end{lemma}

Note that $G$ is almost geometric: distances between vertices
are measured using the $L_\infty$-metric, except that points on the diagonal can be matched for free to
each other if they are not matched with off-diagonal points.
Can this almost-geometric structure speed up computation?
This question motivates our work.

It is possible to simplify the above construction.
We call an edge $uv\in U\times V$ a \emph{skew edge} if $u\in X_0$, $v\in X_0'$ and
$v$ is not the projection of $u$, or if $v\in Y_0$, $u\in Y_0'$ and
$u$ is not the projection of $v$
(skew edges are shown with dashed lines in Figure~\ref{fig:matching-graph}).
\begin{lemma}
\label{lem:skew_lemma}
For both bottleneck and Wasserstein distance, there exists an optimal matching in $(G^q,c^q)$
that does not contain any skew edge.
\end{lemma}
\begin{proof}
Fix an arbitrary matching $M$ and define the matching $M'$ as follows:
For any $uv\in M\cap X_0\times Y_0$, add $uv$ and $u'v'$ to $M'$. For any skew edge $ab'$ of $M$ with $a$ the off-diagonal point (either in $X_0$ or $Y_0$),
add $aa'$ to $M'$. Also add to $M'$ all edges of $M$ of the form $aa'$, where $a$ is an off-diagonal
point. It is easy to see that $M'$ is a perfect matching without skew edges, and its cost
is not worse than the cost of $M$:
indeed, the skew edge $ab'$ got replaced by $aa'$ which is not larger, and the
vertices on the diagonal possibly got rearranged, which has no effect on the cost.
\end{proof}
Lemma~\ref{lem:skew_lemma} implies that removing all skew pairs does not affect
the result of the algorithm, saving roughly a factor of two in the size of the graph.%
\footnote{\libraryname{Dionysus} uses the same simplification.}

We prove another equivalent characterization of the optimal cost
which will be useful in Section~\ref{sec:weighted_wasserstein}:
The previous lemma showed that, conceptually, increasing the weight
of each skew edge to $\infty$ does not affect the cost of an optimal
matching.
We show now that even decreasing the weight of a skew edge $ab'$
to the weight of $aa'$ has no effect on the optimal cost.
Formally, let us define $\tilde{G}=(U\sqcup V,U\times V,\tilde{c})$ 
with a new weight function $\tilde{c}$ as follows:
\begin{equation}
    \label{eq:c_q_tilde}
    \tilde{c}(u,v)=\begin{cases} \|u-v\|_\infty & \text{if}~u\in X_0 ~\text{and}~ v\in Y_0\\
    \|u - u'\|_{\infty} & \text{if}~u\in X_0 ~\text{and}~ v\in X_0'\\
    \|v - v'\|_{\infty} & \text{if}~u\in Y_0' ~\text{and}~ v\in Y_0\\
    0 & \text{otherwise}\end{cases}.
\end{equation}

\begin{lemma}
\label{lem:skew_lemma_tilde}
For both bottleneck and Wasserstein distance, there exists an optimal matching in $\tilde{G}$
that does not contain any skew edge.
\end{lemma}
\begin{proof}
The proof of Lemma~\ref{lem:skew_lemma} carries over word by word.
\end{proof}

\begin{lemma}
\label{lem:diag_lemma}
The weighted graphs $G$ and $\tilde{G}$ 
have the same bottleneck and Wasserstein cost.
\end{lemma}
\begin{proof}
    Let $C$ be the cost for $G$, and $\tilde{C}$ be the cost for
    $\tilde{G}$ with respect to bottleneck or Wasserstein distance.
    Since $\tilde{c}\leq c$ edge-wise, $\tilde{C}\leq C$ is immediate.
    For the opposite direction, fix a matching $\tilde{M}$ that realizes $\tilde{C}$
    and has no skew edge (such a matching exists by Lemma~\ref{lem:skew_lemma_tilde}).
    By the absence of skew edges, the cost $\tilde{M}$ is the same if the cost function
    $\tilde{c}$ is replaced by $c$. This implies $C\leq\tilde{C}$.
\end{proof}

\paragraph{K-d trees.}
K-d trees~\cite{kd-tree} are a classical data structure for near-neighbor search
in Euclidean spaces. The input point set is split into two halves at the median value
of the first coordinates. The process is repeated recursively on the two halves,
cycling through the coordinates used for splitting. Each node of the resulting
tree corresponds to a bounding box of the points in its subtree. The boxes
at any given level are balanced to have roughly the same number of points. Given a query point
$q$, one can find its nearest neighbor (or all neighbors within a given radius)
by traversing the tree. A subtree can be eliminated from the search if the bounding
box of its root node lies farther from the query point than the current
candidate for the nearest neighbor (or the query radius). Although the worst case query performance is
$O(\sqrt{n})$ in the planar case, k-d trees perform well in practice
and are easy to implement. In Section~\ref{sec:bottleneck}
we use the \libraryname{ANN}~\cite{ANN} implementation of k-d trees, changing it
to support the deletion of points. For Section~\ref{sec:wasserstein} we implemented our own
version of k-d trees to support the search for a nearest neighbor with weights.
\paragraph{Experimental setup.}
All experiments in the paper were performed on a server running Debian wheezy, with 32 Intel
Xeon cores clocked at 2.7GHz, with 264 GB of RAM. Only one core was used per
instance in all our experiments.

\begin{figure}[t!]
    \centering
    \begin{subfigure}[t]{0.4\textwidth}
        \centering
        \includegraphics[width=4cm]{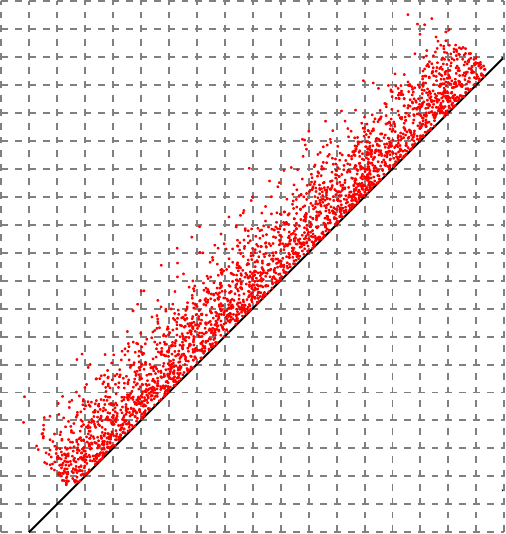}
        \caption{Example of a \dtype{normal} diagram.}
    \end{subfigure}%
    ~
    \begin{subfigure}[t]{0.4\textwidth}
        \centering
        \includegraphics[width=4cm]{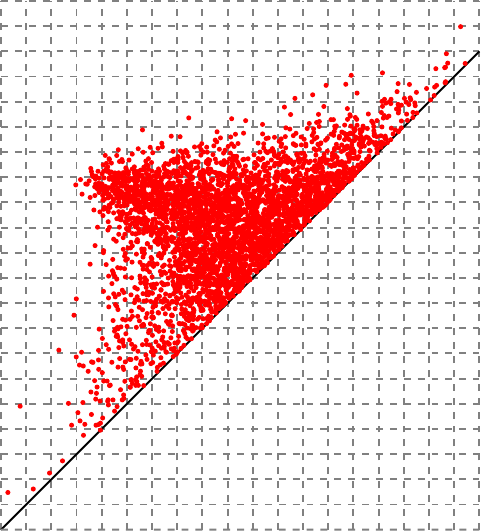}
        \caption{Example of a \dtype{real} diagram.}
    \end{subfigure}
    \caption{Examples of persistence diagrams.}
    \label{fig:example-persistence-diagrams}
\end{figure}

We experimentally compare the performance both on artificially generated diagrams
as well as on realistic diagrams obtained from point cloud data.
For brevity, we restrict the presentation to two classes of instances.
In the first class, we generate pairs of diagrams, each consisting of $n$ points.
The points are  of the form $(a-|b|/2,a+|b|/2)$
where $a$ is drawn uniformly in an interval $[0,s]$,
and $b$ is chosen from a normal distribution $N(0,s)$, with $s=100$.
In this way, the persistence of a point, $|b|$, is normally distributed,
so the point set tends to concentrate near the diagonal.
This matches the behavior of persistence diagrams of realistic data sets, where
points with high persistence are sparse, while the noise present in the data
generates the majority of the points, with small persistence.
For every set of parameters, we generate $10$ independent pairs of diagrams.
We refer to this class of experiments as \dtype{normal} instances (Figure~\ref{fig:example-persistence-diagrams}(a)).

To get a diagram of the second class, we sample a point set $P$ of $n$ points uniformly at random
from either a 4-, or a 9-dimensional unit sphere.
The 1-dimensional persistence diagram of the Vietoris--Rips filtration
of $P$ serves as our input.
We use the \libraryname{Dipha} library\footnote{\url{https://github.com/DIPHA/dipha}} for the generation
of these instances. Note that persistence diagrams generated in this way
have different numbers of points.
We refer to this class of experiments as \dtype{real} instances (Figure~\ref{fig:example-persistence-diagrams}(b)).
For each set of parameters (sphere dimension and number of points sampled),
we have generated $6$ test instances and
computed pairwise distances between all ${6\choose 2} = 15$ pairs.

Our plots show the average running times and the standard deviation as error bars.
For the \dtype{real} class, the $x$-axis is labelled with the number of points
sampled from the sphere, not with the size of the diagram.
The size of the persistence diagrams, however, depends
practically linearly on the number of sample points, with a constant factor
that grows with dimension:
the largest instance for dimension $9$ is a diagram with $5762$ points,
while for dimension $4$ the largest diagram is of size $1679$.

Our experiments cover many other cases.
We have tested various choices of $s$, the scaling parameter in the \dtype{normal} class,
and of the sphere dimension in the \dtype{real} class.
We have also tried different ways of generating diagrams, for instance, by choosing $n$ points
uniformly at random in the square  $[0,s]\times[0,s]$, above the diagonal.
In all these cases, we encountered the same qualitative difference between the tested algorithms
as for the two representative cases discussed in this paper.

\section{Bottleneck matchings}
\label{sec:bottleneck}
Our approach follows closely the work of Efrat et al.~\cite{eik-geometry},
based on the following simple observation.
Let $G[r]$ be the subgraph of $G$ that contains the edges with weight at most $r$.
The bottleneck distance of $G$ is the minimal value $r$ such that $G[r]$
contains a perfect matching.
Since the bottleneck cost for $G$ must be equal to the weight of one of the edges,
we can find it exactly by combining a test for a perfect
matching with a binary search on the edge weights.

\paragraph{The algorithm by Hopcroft and Karp.}
 Efrat et al.~modify the algorithm by Hopcroft and Karp~\cite{hk-algorithm}
to find a maximum matching. We briefly summarize the Hopcroft--Karp algorithm;
\cite{eik-geometry} provides an extended review.
For a given graph $G[r]$, the algorithm computes a maximum matching, i.e.,
a matching of maximal cardinality.
$G[r]$, with $2n$ vertices, has a perfect matching if and only if its maximum
matching has $n$ edges.

The algorithm maintains an initially empty matching $M$ and looks for an \emph{augmenting path},
i.e., a path in $G[r]$ that alternates between edges inside and outside of $M$,
with the first and the last edge not in $M$.
Switching the state of all edges in an augmenting path (inserting or removing them
from $M$) \emph{augments} the matching, increasing its size by one.

The algorithm detects several vertex-disjoint augmenting paths at once.
It computes a \emph{layer subgraph} of $G[r]$,
from which it reads off the vertex-disjoint augmenting paths.
Both the construction of the layer subgraph and the search for augmenting paths
are realized through a graph traversal in $G[r]$ in $O(m)$ time, where $m$ is the number of edges.
Having identified augmenting paths, the algorithm augments the matching and
starts over, repeating the search until all vertices are matched or no
augmenting path can be found.
As shown in~\cite{hk-algorithm}, the algorithm terminates after $O(\sqrt{n})$ rounds,
yielding a running time of $O(m\sqrt{n})=O(n^{2.5})$.

\paragraph{Geometry helps.}
The crucial observation of Efrat et al.~is that for a geometric graph $G[r]$,
the layer subgraph does not have to be constructed explicitly.
Instead one may use a near-neighbor search data structure,
denoted by $\nnd_r(S)$, which stores a point set $S$ and a radius $r$.
It must answer queries of the form: given a point $q\in\R^2$, return a point
$s \in S$ such that $d(q,s)\leq r$. $\nnd_r(S)$ must support deletions of points in $S$.
As the authors show, if $T(|S|)$ is an upper bound for the cost of one operation in $\nnd_r(S)$,
the algorithm by Hopcroft and Karp runs in $O(n^{1.5}T(n))$ time for a graph with $2n$ vertices.
For the planar case, Efrat et al.~show that one can construct such a data structure
(for any $L_p$-metric) in $O(n\log n)$ preprocessing time, with $T(n)=O(\log n)$
time per operation.
Thus, the execution of the Hopcroft--Karp algorithm costs only $O(n^{1.5}\log n)$.

Naively sorting the edge weights and binary searching for the value of $r$ takes $O(n^2\log n)$ time.
But this running time would dominate the improved
Hopcroft--Karp algorithm. In order to improve the complexity of the edge search,
the authors use an approach, attributed to Chew and Kedem~\cite{ck-improvements},
for efficient $k$-th distance selection for a bi-chromatic point set under the $L_\infty$-distance;
see \cite[Sec.6.2.2]{eik-geometry} for details.

With this technique,
the computation of a maximum matching dominates the cost of finding the $k$-th
largest distance, giving the runtime complexity of $O(n^{1.5}\log^2 n)$ for computing the bottleneck matching.
Using further optimizations~\cite[Sec.5.3]{eik-geometry},
they obtain a running time of $O(n^{1.5}\log n)$ for geometric graphs in $\R^2$
with the $L_\infty$-metric.

It is not hard to see that the analysis carries over to the case of persistence
diagrams (also mentioned in~\cite[p.196]{eh-computational}).
Let $G_1=(U\sqcup V,U\times V)$ be the graph defined in Lemma~\ref{lem:reduction_lemma}.
In the algorithm, $\nnd_r(S)$ is initialized with the points in $V$,
which are subsequently removed from it.
We additionally maintain a set $S'$ of diagonal points contained in $S$. When
the algorithm queries a near neighbor of a diagonal point of $U$, we return one
of the diagonal points from $S'$ in constant time, if $S'$ is not empty.  The
overhead of maintaining $S'$ is negligible. We summarize:

\begin{theorem}
The bottleneck distance of two persistence diagrams can be computed in $O(n^{1.5}\log n)$.
\end{theorem}

\paragraph{Our approach.}
Our implementation follows the basic structure of Efrat et al., reducing the construction
of layered subgraphs to operations on a near-neighbor data-structure $\nnd_r(S)$.
But instead of the rather involved data structure proposed by the authors,
we use a simpler alternative: we construct a k-d tree for $S$.
When searching for a point at most $r$ away from a query point $q$, we traverse
the k-d tree, pruning from the search the subtrees whose enclosing box is
further away from the query than the current best candidate.
When a point is removed from $S$, we mark it as removed in the k-d tree;
in particular, we do not rebalance the tree after a removal.
We also keep track of how many points remain in each subtree,
so that we can prune empty subtrees from the subsequent searches.
The running time per search query can be bounded by $O(\sqrt{n})$
per query, with $n$ the number of points originally stored in the search tree.
We remark that using range trees~\cite{dutch}, the worst-case complexity
could be further reduced to $O(\log n)$.

Initial tests showed that the naive approach of precomputing and sorting all
distances for the binary search dominates the running time in practice.
Instead of implementing the asymptotically fast but complicated approach
of Efrat et al., we compute a $\delta$-approximation of the bottleneck
distance, which we can then post-process to compute the exact answer.
Let $d_{\max}$ denote the maximal $L_\infty$-distance between a point in $U$ and a point in $V$ in $G_1$.
First, we compute, in linear time, a $3$-approximation of $d_{\max}$ as follows.
We pick an arbitrary point in $U$, find its farthest point $v_0 \in V$,
and find a point $u_0 \in U$ farthest from $v_0$. Then, $\|u_0-v_0\|_\infty\leq d_{\max}\leq 3 \|u_0-v_0\|_\infty$
(from the triangle inequality).
Setting $t=3 \|u_0-v_0\|_\infty$, the exact bottleneck distance $o$ must be in $[0, t]$
and we perform a binary search on $[0, t]$ until we find an interval $(a,b]$ that
satisfies $(b-a)<\delta\cdot a$.  We return $b$ as the approximation.
It is easy to see that $b \in [o, (1+\delta)o)$.

At each iteration of the binary search, we reuse the maximum matching constructed before
(if the true distance is below the midpoint of the current interval $(a, b]$, we remove
edges whose weight is greater than $( a + b) / 2$, otherwise the whole matching can be kept).

\begin{figure}[t]
\centering
\includegraphics[width=3cm]{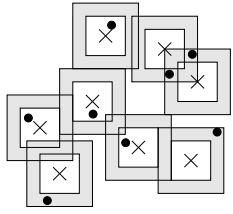}
\caption{Illustration of the exact computation step: the exact bottleneck distance
must be realized by a point in $B$ (circles) in an annulus around $A$ (crosses).
The width of the annulus is determined by the approximation quality.
In this example, there are 6 candidate pairs.}
\label{fig:certification}
\end{figure}

To get the exact answer, we find pairs in $U\times V$ whose distance is in
the approximation interval, $(a,b]$.
For such a pair $(u,v)$,
$v$ lies in an $L_\infty$-annulus around $u$ with inner radius $a$ and outer radius $b$.
So we find for every $u\in U$ the points of $V$ in the
corresponding annulus and take the union of all such pairs as the candidate set.
In the example in Figure \ref{fig:certification}, points in $U$ are drawn as crosses, points in $V$ as circles,
and there are 6 candidate pairs.

We compute the candidate pairs with similar techniques as used for range trees~\cite{dutch}.
Specifically, we identify all pairs $(u,v)$ whose $x$-coordinate difference lies
in $(a,b]$. We can compute the set $C_x$ of such pairs
in $O(n\log n + |C_x|)$ time by sorting $U$ and $V$ by $x$-coordinates.
For each pair $(u,v)$ in $C_x$,
we check in constant time whether $\|u-v\|_\infty\in (a,b]$ and remove the pair otherwise.
We then repeat the same procedure using the $y$-coordinates.
To compute the exact bottleneck distance, we perform binary search on the vector
of candidate distances.

Let $c$ denote the number of candidate pairs.
The complexity of our procedure is not output-sensitive in $c$
because $|C_x|+|C_y|$ can be larger than $c$ --- so too many pairs might be considered.
Nevertheless, we expect that when using a sufficiently good initial
approximation, both $|C_x|+|C_y|$ and $c$ are small, so our method will be
fast in practice.

\paragraph{Experiments.}
We compare the geometric and non-geometric bottleneck matching algorithms.
We set $\delta=0.01$ and compute the approximate bottleneck
distance to the relative precision of $\delta$, using k-d trees
for the geometric version and constructing the layered
graph combinatorially in the non-geometric version.
Figure~\ref{fig:bottleneck:geometric_non_geometric} shows the results for
\dtype{normal} and \dtype{real} instances.
We observe that the geometric version scales significantly better, and runs
faster by a factor of roughly $10$ for the largest displayed \dtype{normal}
instance with $25000$ points per diagram.
We remark that the memory consumption of the geometric and non-geometric versions
both scale linearly, and the geometric version is larger by a factor
of roughly $4$ throughout. For $25000$ points,
about 60MB of memory is required.

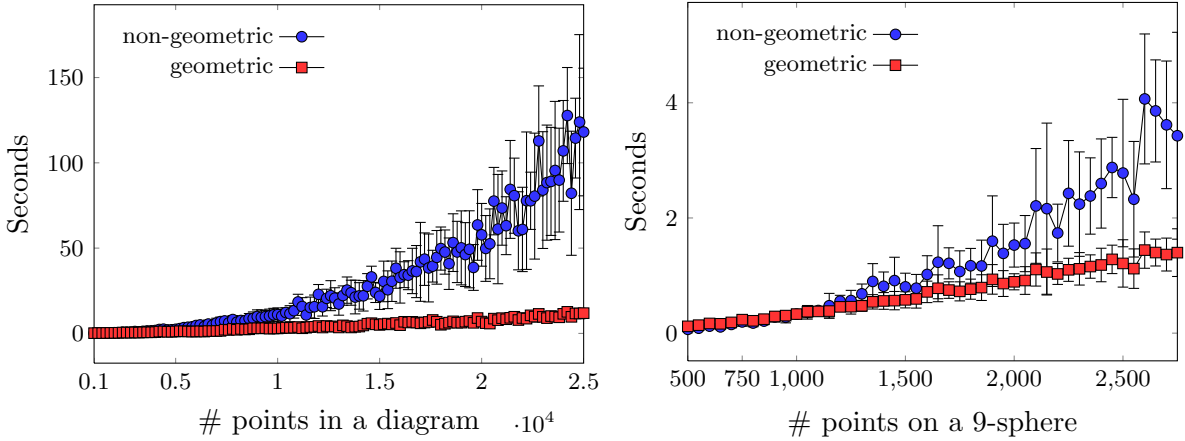
\begin{figure*}[t]
    \centering
    \parbox{0.49\textwidth}{
    \begin{tikzpicture}
    \begin{axis}[
                    xmin=1000,xmax=25000,
                    xtick={1000,5000,10000,15000,20000,25000},
                    height=2.5in,
                    width=0.49\textwidth,
                    xlabel={\# points in a diagram},
                    ylabel={Seconds},
                    legend pos=north west,
                ]
        \addaveraged[nongeommark]{test_nongeom_0.01_normal_distr_1000_1000.txt}
        \addlegendentry{non-geometric}
        \addaveraged[geommark]{log_test_normal_distr_1000_1000_bottleneck_final_0.01.txt}
        \addlegendentry{geometric}
    \end{axis}
    \end{tikzpicture}
    }
    \parbox{0.49\textwidth}{
    \begin{tikzpicture}
    \begin{axis}[
                    xmin=500,xmax=2750,
                    xtick={500,750,1000,1500,2000,2500},
                    height=2.5in,
                    width=0.49\textwidth,
                    xlabel={\# points on a 9-sphere},
                    ylabel={Seconds},
                    legend pos=north west,
                ]
        \addaveraged[nongeommark]{nongeom-approx-new-9.txt}
        \addlegendentry{non-geometric}

        \addaveraged[geommark]{geom-approx-new-dim-9.txt}
        \addlegendentry{geometric}
    \end{axis}
    \end{tikzpicture}
    }
    \caption{Running times of the bottleneck distance computation on
            \dtype{normal} data (left) and \dtype{real} data (right) for varying
            number of points.}
    \label{fig:bottleneck:geometric_non_geometric}
\end{figure*}

We used linear regression to fit curves of the form $cn^\alpha$
to the plots of Figure~\ref{fig:bottleneck:geometric_non_geometric} (left).
For the non-geometric version, the best fit appeared for $\alpha=2.3$,
roughly matching the asymptotic bound of Hopcroft--Karp.
For the geometric version, we get the best fit for $\alpha=1.4$;
this shows that despite the pessimistic worst-case complexity,
the algorithm tends to follow the improved geometric bound on practical instances.

The above experiment does not include the post-processing step of computing the
exact bottleneck distance. We test the geometric version above that yields a
$1\%$ approximation against the variant that also computes the exact distance
from the initial approximation, as explained earlier in this section. Our
experiments show that the running time of the post-processing step is about half
of the time needed to get the approximation.  Although there is some variance in
the ratio, it appears that the post-processing does not worsen the performance
by more than a factor of two.

Figure~\ref{fig:bottleneck:comp_with_dionysus} compares our
exact (geometric) bottleneck algorithm with \libraryname{Dionysus},
the only publicly available implementation for computing bottleneck distance between persistence
diagrams. \libraryname{Dionysus} simply sorts the edge distances in increasing order and performs a binary
search, building the graphs $G[r]$
and calling the Edmonds matching algorithm~\cite{edmonds-paths}
from the \libraryname{Boost} library to check for a perfect matching in $G[r]$.
Already for diagrams of $2800$ points, our speed-up exceeds a factor of $400$.

\begin{figure*}[t]
    \centering
    \parbox{0.49\textwidth}{
    \begin{tikzpicture}
    \begin{axis}[
                    ymode=log,
                    xmin=1000,xmax=2800,
                    ymax=1000,
                    xtick={1000,1500,2000,2800},
                    height=2.5in,
                    width=0.49\textwidth,
                    xlabel={\# points in a diagram},
                    ylabel={Seconds},
                    legend pos=north west,
                ]
        \addplot+[dionmark] table[x index=1, y index=15] {normal_distr_all_with_dion.csv};
        \addlegendentry{Dionysus}
        \addplot+[geommark] table[x index=1, y index=3] {normal_distr_all_with_dion.csv};
        \addlegendentry{geometric}
    \end{axis}
    \end{tikzpicture}
    }
    \parbox{0.49\textwidth}{
    \begin{tikzpicture}
    \begin{axis}[
                    ymode=log,
                    xmin=500,xmax=2750,
                    height=2.5in,
                    width=0.49\textwidth,
                    xlabel={\# points on a 9-sphere},
                    ylabel={Seconds},
                    legend pos=north west,
                ]
        \addaveraged[dionmark]{dion-new-dim-9.txt}
        \addlegendentry{Dionysus} 

        \addaveraged[geommark]{geom-exact-new-dim-9.txt}
        \addlegendentry{geometric} 
    \end{axis}
    \end{tikzpicture}
    }
    \caption{Comparison of our exact geometric bottleneck algorithm with
            \libraryname{Dionysus} for \dtype{normal} (left) and \dtype{real}
            (right) input.}
    \label{fig:bottleneck:comp_with_dionysus}
\end{figure*}
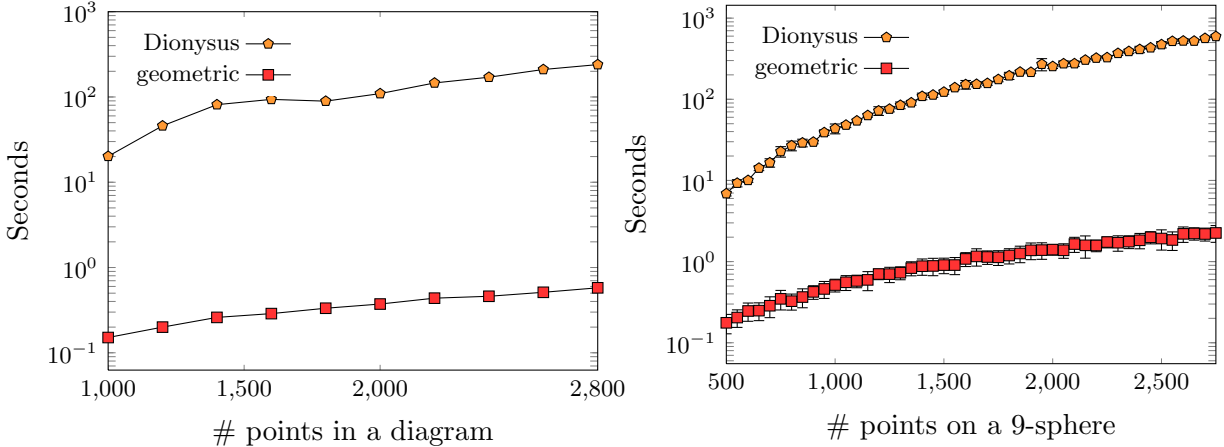

\section{Wasserstein matchings}
\label{sec:wasserstein}

We now fix $q\geq 1$ and describe an algorithm for computing the $q$-Wasserstein cost 
of a weighted graph $(U\sqcup V,E,w)$. Recall from Proposition~\ref{prop:q_to_1_reduction}
that we can restrict to the $1$-Wasserstein case by switching to the cost function $w^q$.
Moreover, we assume that $U=\{u_1,\ldots,u_n\}$ and $V=\{v_1,\ldots,v_n\}$ 
are finite sets, and we identify the elements with their indices.

\paragraph{Auction algorithm.}
The \emph{auction algorithm} of Bertsekas~\cite{bertsekas-distributed} is an
asymmetric approach to find a perfect matching in a weighted graph
that maximizes the sum of its edge weights.
One half of the bipartite graph is treated as ``bidders'',
the second half as ``objects.'' Initially, each object $j$ is assigned zero price,
$p_j = 0$, and each bidder $i$ extracts a certain benefit, $b_{ij}$, from object $j$.
Since we are interested in the minimum cost matching, we use the negation of the
edge weight as the bidder--object benefits, that is,
$b_{ij} = -w^q(i,j)$.
If the edge $(i,j)$ is not in the graph, $b_{ij} = -\infty$.
The auction algorithm maintains a (partial) matching $M$,
which is empty initially.
When $M$ becomes perfect, the algorithm stops.
During the execution of the algorithm,
matched bidders in $M$ are called \emph{assigned} (to an object),
and unmatched bidders are \emph{unassigned}.

The auction proceeds iteratively. In each iteration, one unassigned bidder $i$
chooses an object $j$ with the maximum value, defined as the benefit minus the
current price of the object, $v_{ij} = (b_{ij} - p_j$).
Object $j$ is assigned to the bidder; if it was assigned before,
the previous owner becomes unassigned.
Let $\Delta p_{ij}$ denote the difference of $v_{ij}$ and the value
of the second best object for bidder $i$;
$\Delta p_{ij}$ can be zero.
The price of object $j$ increases by $\Delta p_{ij}+\epsilon$,
where $\epsilon$ is a small constant needed to avoid infinite loops in cases where 
two bidders extract the same value from two objects. Without
$\epsilon$, the two could keep stealing the same object
from each other without increasing its price.

Our variant of the algorithm is called \emph{Gauss--Seidel auction}:
an iteration consists of only one bid, which is always satisfied.
An alternative, called the \emph{Jacobi auction}, proceeds by letting each
unassigned bidder place a bid in every iteration. If several bidders
want the same object, it is assigned to the bidder who offers the highest
price increment, $\Delta p_{ij} + \epsilon$.
The Jacobi auction, which was used in the ALENEX version of this paper~\cite{kmn-geometry},
has a drawback if many objects provide the same value to many bidders.
In that case, it may happen that all of these bidders bid for the same object
in one iteration, and all but one of them remain unassigned.
Since a Jacobi iteration is more expensive than a Gauss--Seidel iteration,
this may result in worse performance. Indeed, our experiments show
that switching to Gauss--Seidel auction improves the runtime by an
order of magnitude.

How do we choose $\epsilon$?
Small values give a better approximation of the exact answer; on the other hand,
the algorithm converges faster for large values of $\eps$.
Bertsekas suggests \textit{$\eps$-scaling} to overcome this problem:
running several rounds of the auction algorithm with decreasing values of
$\eps$, using prices from the previous round, but an empty matching, as an
initialization for the next round.  Following the recommendation of Bertsekas and Casta{\~n}on~\cite{bc-auction-parallel},
we initialize $\eps$ with the maximum edge cost 
divided by $4$ and divide $\eps$ by $5$ when starting a new round.

Iterating this procedure long enough would eventually yield the exact
Wasserstein distance~\cite{bertsekas-distributed}; however,
the number of rounds of $\epsilon$-scaling would in general be too high
for many practical problems. Instead, we use a termination condition
that guarantees a relative approximation of the exact value. We fix
some approximation parameter $\delta\in(0,1)$. 
After finishing a round of the auction algorithm
for $q$-Wasserstein matching
for some value $\epsilon>0$, let $d:=d_\epsilon$ be the $q$-th root of 
the cost of the obtained matching.
We stop if $d$ satisfies
\begin{equation}
    \label{eqn:auction_termination}
d^q\leq (1+\delta)^q(d^q-n\eps),
\end{equation}
and return $d$ as the result of the algorithm. We summarize the auction
in Algorithm~\ref{alg:auction}.
\noindent\begin{lemma}
    \label{lemma:rel_error_auction}
The return value $d$ of the algorithm satisfies
\[d\in[\opt,(1+\delta)\opt),\]
where $\opt$ denotes the exact $q$-Wasserstein distance.
\end{lemma}
\begin{proof}
Because we raise all edge costs to the $q$-th power, 
the matching minimizing the sum of the edge costs has a cost of $\opt^q$.
Let $d^q$ be the cost of the matching computed by the auction algorithm,
after the last round of $\eps$-scaling, for a fixed $\eps$.
By the properties of the auction algorithm~(\cite{bertsekas1988auction}, Proposition 1), it holds (after every round)
that
\[\opt^q\leq d^q \leq \opt^q+n\eps.\]
Taking the $q$-th root yields the first inequality immediately. 
For the second inequality, note that
\[(1+\delta)^q\opt^q \geq (1+\delta)^q(d^q-n\eps)\geq d^q,\]
where the last inequality follows from the termination condition of the
algorithm. Taking the $q$-th root on both sides yields the result. 
\end{proof}

\begin{algorithm}
\DontPrintSemicolon
\KwIn{Two persistence diagrams $X$, $Y$ with $|X|,|Y|\leq n$, $q \geq 1$, $\delta > 0$ (maximal relative error) }
\KwOut{$\delta$-approximate $q$-Wasserstein distance $W_q(X, Y)$}
Initialize $d\gets 0$ and $\eps\gets \frac{5}{4}\cdot \left(\text{max. edge length}\right)$\;

\While{\ $d^q>(1+\delta)^q(d^q-n\eps)$\ }{
  $\eps \gets \eps / 5$\;
  Let $M$ be an empty matching\;
    \While{\mbox{there exists some unassigned bidder $i$}}{
        Find the best object $j$ with value $v_{ij}$ 
        and the second best object $k$ with value $v_{ik}$ for $i$\;
        Assign $j$ to $i$ in $M$ and increase the price of $j$ 
        by $(v_{ij}-v_{ik})+\eps$\;
    }  
  $d\gets$ $q$-th root of the cost of the (perfect) matching $M$\;
}
\Return{d}\;
\caption{{\sc Auction algorithm}}
\label{alg:auction}
\end{algorithm}

\paragraph{Bidding.}
The computational crux of the algorithm is for a bidder to determine the object of maximum value and the price increase.
The brute-force approach is for each bidder to do an exhaustive search over all
objects. Doing so requires linear running time per iteration.
But let us
consider what the search actually entails.
Bidder $i$ must find the two objects with highest and second-highest
$v_{ij}$ values.
Recall $v_{ij} = b_{ij} - p_j = -w^q(i,j) - p_j$,
and maximizing this quantity for a fixed $i$ is equivalent to
minimizing $w^q(i,j) + p_j$.

The first way to quickly find these objects uses lazy heaps.
Each bidder keeps all the objects in a heap, ordered by their value.
We also maintain a list of all the price changes (for any object), as well
as a record for each bidder of the last time its heap was updated.
Before making a choice, a bidder updates the values of all the objects in its
heap that changed prices since the last time the heap was updated.
The bidder then selects the two objects with the maximum value.
We note that this approach uses quadratic space,
since each bidder keeps a record of each object.

The second way to accelerate the search for the best object uses geometry and
requires only linear space.
Initially, when all the prices are zero, we can find the two best objects
by performing the proximity search in a k-d tree.
But we need to augment the k-d tree to take increasing prices into account.
We do so by storing the price of each point as its weight in the
k-d tree. At each internal node of the tree we record the minimum weight of any
node in its subtree. When searching, we prune subtrees if the $q$-th power of the
distance from the query point to the box containing all of the subtree's points,
plus the minimum weight in the subtree, exceeds the current second
best candidate.

Once a bidder selects the best object, it increases its price. We adjust the
subtree weights in the k-d tree by increasing the chosen object's
weight and updating the weights on the path to the root.
If the minimum weight does not change at some node on the path, we interrupt the
traversal.

The case of persistence diagrams requires special care.
We can distinguish between \emph{diagonal} and \emph{off-diagonal} bidders and
objects.
Diagonal bidders should bid for only one off-diagonal object, according to Lemma~\ref{lem:skew_lemma}.
Since the distance between diagonal points is $0$, the value of a diagonal object $j$
for a digonal bidder $i$ is just the opposite of its price, $v_{i,j} = -p_j$,
and we keep all diagonal objects in a heap ordered by the price.
When a diagonal bidder needs to find the best two objects,
it selects the top two elements of the heap and compares them with the only off-diagonal
object to which it can be assigned.

On the other hand, off-diagonal bidders can bid for every
off-diagonal object and only for one diagonal object (its projection).
We use one global k-d tree to get the best two off-diagonal objects and then
compare their values for the bidder with the value of bidder's projection,
so only off-diagonal objects are stored in the k-d tree.

\begin{figure}[t]
    \centering
    \begin{tikzpicture}
    \begin{axis}[
                    ymode=log,
                    ymax=10000000,
                    xtick={1000,2000,3000,4000,5000,6000,7000,8000,9000,10000},
                    xmin=1000,xmax=10000,
                    height=2.5in,
                    width=.49\textwidth,
                    xlabel={\# points in a diagram},
                    ylabel={Kilobytes},
                    legend pos=north west,
                ]
        \addaveraged[nongeommark]{gs-lazy-heap-normal-memory-100-100.txt}
        \addlegendentry{non-geom.}

        \addaveraged[geommark]{gs-kdtree-normal-memory-100-100.txt}
        \addlegendentry{geom.}
    \end{axis}
    \end{tikzpicture}

    \caption{Comparison of memory consumption of geometric and non-geometric versions of auction algorithm on \dtype{normal} instances.}
\label{fig:memory-consumption-wasserstein}
\end{figure}
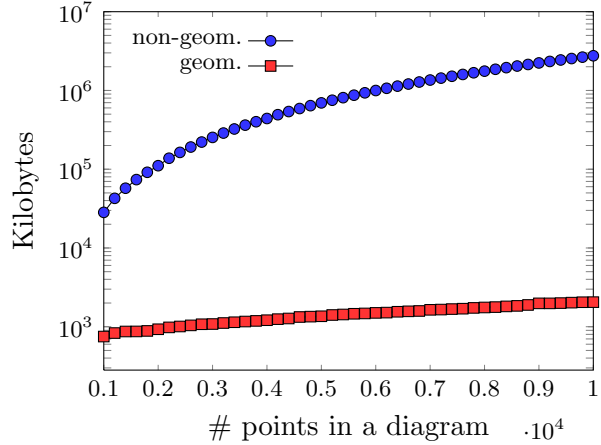

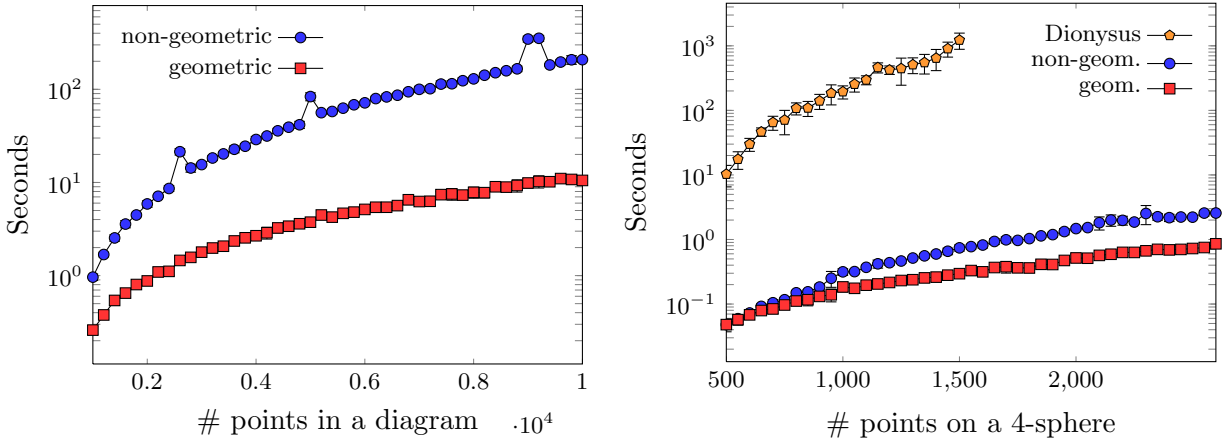
\begin{figure*}[t]
    \centering
    \parbox{0.49\textwidth}{
    \begin{tikzpicture}
    \begin{axis}[
                    ymode=log,
                    xtick={2000,4000,6000,8000,10000},
                    xmin=1000,xmax=10000,
                    height=2.5in,
                    width=.49\textwidth,
                    xlabel={\# points in a diagram},
                    ylabel={Seconds},
                    legend pos=north west,
                ]

        \addaveraged[nongeommark]{gs-lazy-heap-normal-100-100.txt}
        \addlegendentry{non-geometric}

        \addaveraged[geommark]{gs-kdtree-normal-100-100.txt}
        \addlegendentry{geometric}
    \end{axis}
    \end{tikzpicture}
    }
    \parbox{0.49\textwidth}{
    \begin{tikzpicture}
    \begin{axis}[
                    ymode=log,
                    xmin=500,xmax=2600,
                    xtick={500,1000,1500,2000},
                    height=2.5in,
                    width=0.49\textwidth,
                    xlabel={\# points on a 4-sphere},
                    ylabel={Seconds},
                    legend pos=north east,
                ]

        \addaveraged[dionmark]{dion-4sphere-new.txt}
        \addlegendentry{Dionysus}

        \addaveraged[nongeommark]{gs-heap-4sphere-new.txt}
        \addlegendentry{non-geom.}

        \addaveraged[geommark]{gs-kdtree-4sphere-new.txt}
        \addlegendentry{geom.}
    \end{axis}
    \end{tikzpicture}
    }

    \caption{Comparison of non-geometric and geometric variants of the auction
             algorithm on \dtype{normal} (left) and \dtype{real} (right) input,
             also with Dionysus on the \dtype{real} input.}
    \label{fig:wasserstein:geometric_non_geometric}
\end{figure*}

\paragraph{Experiments.}
Figure~\ref{fig:wasserstein:geometric_non_geometric} illustrates the running times of the
auction algorithm on the \dtype{normal} data, using lazy heaps and k-d trees.
In both cases, we compute a relative $0.01$-approximation.
The advantage of using geometry is evident: the algorithm
is faster by roughly a factor of $4$ for diagrams with $1000$ points,
and the factor becomes close to $20$ for diagrams with $10000$ points.
We used linear regression to empirically estimate the complexity, and the geometric
algorithm runs in $O(n^{1.6})$, while for the non-geometric algorithm the
estimated complexity is super-quadratic, $O(n^{2.3})$.
The non-geometric version only shows competitive running times because of the described
optimization with lazy heaps. This results in a severe increase in memory consumption,
as displayed in Figure~\ref{fig:memory-consumption-wasserstein}.

Again, we compare our geometric approach with
\libraryname{Dionysus}, which uses John Weaver's
implementation\footnote{\url{http://saebyn.info/2007/05/22/munkres-code-v2/}}
of the Hungarian algorithm~\cite{munkres-algorithms}.
Figure~\ref{fig:wasserstein:geometric_non_geometric} (right) shows the results for \dtype{real} instances.
The speed-up of our approach increases from a factor of $50$ for small instances to a factor of about $400$
for larger instances. For the \dtype{normal} data sets, the speed-up already exceeds a factor of $1000$
for diagrams of $1000$ points; we therefore omit a plot.

We emphasize that our test is slightly unfair, as it compares the exact algorithm
from \libraryname{Dionysus} with the $0.01$-approximation
provided by our implementation.
While such an approximation suffices for many applications
in topological data analysis, the question remains how much overhead
would be caused by an exact version of the auction algorithm.
A naive approach to get the exact result is to rescale the
input to integer coordinates and to choose $\eps$ such that
the approximation error is smaller than $1$.
We plan to investigate different possibilities to compute the exact distance
more efficiently.

\section{Wasserstein matchings for repeated points}
\label{sec:weighted_wasserstein}

For a weighted, complete, bipartite graph $G=(U\sqcup V,U\times V,w)$,
we call two vertices $u_1,u_2\in U$ \emph{identical}
if for all $v\in V$, $w(u_1,v)=w(u_2,v)$.
A pair of identical vertices in $V$ is defined symmetrically.
If $G$ is a geometric graph, two points with coinciding locations
are identical.
In the context of persistence diagrams,
this situation is common in applications, where the
range of possible scales on which features appear and disappear
is often discretized. The discretization places all points of the persistence
diagram on a finite grid.
For a fixed discretization of a fixed range, more and more identical points
appear as the data size grows.
This raises the question whether diagrams with many identical points
can be handled more efficiently.

\paragraph{Auction with integer masses.}
We use a variant of the auction algorithm~\cite{bc-auction}, which we explain next.
The input consists of two sets $U$ and $V$ of \emph{multi-points},
each given by its coordinates and integer multiplicity $m\geq 1$;
a multi-point represents $m$ identical points at the given location.
For brevity, we refer to the multiplicity as {\emph{mass}}.
The total mass of both sets is the same.
In analogy to the auction algorithm
from Section~\ref{sec:wasserstein}, we refer to the
elements of the respective sets as \emph{multi-bidders}
and \emph{multi-objects}. 
The elements of a multi-object
are not, in general, assigned to the same multi-bidder;
their prices can also differ. However, if two elements of a multi-object
are assigned to one multi-bidder, the algorithm guarantees that their prices are equal.
The algorithm decomposes a multi-object into \emph{slices}, where each slice
represents a fraction of the multi-object that is currently
not distinguished by the algorithm.
Formally, a slice is a four-tuple $(j, m_{i,j}, p_{i,j}, i)$ identifying the multi-object $j$ it belongs to,
the mass of the slice $m_{i,j}$, its price $p_{i,j}$, and the multi-bidder $i$
that it is currently assigned. The decomposition of multi-objects into slices defines an \emph{assignment}, which can be interpreted as a matching $M$ in the original graph (see Fig. \ref{fig:assignment_to_matching}):
A slice $(j,m_{i,j},p_{i,j},i)$ corresponds to $m_{i,j}$ edges in $M$ from $m_{i,j}$ elements of the multi-bidder $i$ to $m_{i,j}$ elements of the multi-object $j$
(hereby interpreting multi-bidders and multi-objects as sets of identical bidders/objects). 
Unassigned slices correspond to unmatched vertices. We call an assignment \emph{perfect} if the induced matching is perfect, 
and the \emph{cost} of the assignment is the cost of the corresponding matching.

The auction with integer masses is a procedure converging to an assignment with minimal cost. It uses the same high-level structure
as the 
auction described in Section~\ref{sec:wasserstein},
which we will refer to as the \emph{standard auction}.
It employs \mbox{$\epsilon$-scaling} with the same choices of parameters.
One round of $\epsilon$-scaling maintains an assignment 
and runs until the assignment is perfect, that is, all multi-bidders
are fully assigned to multi-objects. Every round proceeds in iterations.
In each iteration, one multi-bidder with unassigned mass is selected at random.
It acquires enough slices (possibly taking them away from other multi-bidders)
to assign all its missing mass and increases the prices of these slices.

\begin{figure}[t!]
        \centering
        \includegraphics[width=7cm]{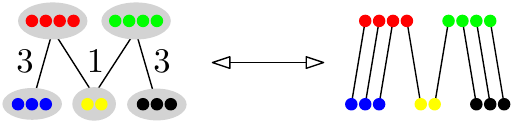}
        \caption{Correspondence between assignments and matchings. On the left-hand side there are two multi-bidders, each of mass 4, and 4 slices with masses 3, 3, 1, 1. A corresponding perfect matching is on the right-hand side. }
    \label{fig:assignment_to_matching}
\end{figure}

Specifically, an iteration proceeds as follows. We fix
a multi-bidder with some unassigned mass  $u\geq 1$, and
let $s_1, \dots, s_t$ be the slices assigned to it.
Conceptually, the algorithm
takes all possible slices except for $s_1, \dots, s_t$ and sorts them by their
value to the multi-bidder in decreasing order. We denote the sorted slices by
$s_{t+1},\ldots,s_N$; let $v_i$ denote the value of $s_i$ to the multi-bidder.

The multi-bidder takes the first $k$ slices $s_{t+1}, \dots, s_{t+k}$
such that their total mass $m$ is at least $u$.
If $m>u$, we split the ``leftover'' slice from $s_{t+k}$
whose mass is $m - u$ and whose price and owner remain unchanged;
we denote this newly created slice as $\tilde{s}_{t+k}$.
Now, the total mass of the slices $s_{t+1},\ldots,s_{t+k}$ is exactly $u$,
and we assign them to the multi-bidder.

Next, we increase the prices of all $t+k$ slices assigned to the multi-bidder.
Let $s_l$ with $l \geq t+k$ be a slice determined as follows:
if the slices $s_1,\ldots,s_{t+k}$ belong to at least two different multi-objects,
$s_l$ is the slice containing the $(m+1)\mbox{-st}$ unit of mass, that is,
$s_l$ is set to $s_{t+k+1}$ if we did not split the leftover slice, and
to $\tilde{s}_{t+k}$ otherwise. If all the $t+k$ slices are of a single multi-object,
then $s_l$ is defined to be the first slice among $s_{t+k+1},\dots,s_N$ that belongs to a different multi-object.
Let $v_l$ be the value of $s_l$ to the multi-bidder. We increase the prices of the
slices $\{ s_i \}_{1 \leq i \leq t+k}$ by $v_i - v_l + \epsilon$ to make them
as valuable to the multi-bidder as the slice $s_l$, up to $\epsilon$.

The original paper that presents this approach~\cite{bc-auction}
describes the Jacobi version of the algorithm, i.e., all bidders with unassigned
mass submit bids in one iteration, and the mass goes to the bidder who offered
the highest bid. The above description is the Gauss--Seidel variant of the same algorithm,
and it is straightforward to verify that the same proof of correctness
works for it, too. From the discussion in \cite{bc-auction}, it follows
that one can use the same formula as for the standard auction 
to estimate the relative error of the
matching obtained after each round of $\epsilon$-scaling.
Therefore, we can use the same termination condition as in~\eqref{eqn:auction_termination}
and the proof of Lemma~\ref{lemma:rel_error_auction} carries over.
We refer to \cite{bc-auction} for further details.

\paragraph{Diagonal points.}
Let $X_0$ and $Y_0$ be the of the off-diagonal
mult-points of two persistence diagrams.
Recall that for the computation of the $q$-Wasserstein distance,
we introduce the projection sets $Y_0'$ and $X_0'$
(also as a set of multi-points with masses inherited from their pre-images)
and set $X:=X_0\cup Y_0'$ and $Y:=Y_0\cup X_0'$,
which are sets of multi-points with equal total mass.
We can run the auction with integer masses, using the cost function $c^q$,
with $c$ as in~\eqref{eq:c_q}, and return the $q$-th root of the obtained
cost as our result.
However, we get a major improvement from using the modified cost function
$\tilde{c}$, defined in ~\eqref{eq:c_q_tilde}.
The modified function decreases the costs of all skew edges; accordingly,
$\tilde{c}^q$ treats all points in $X_0'$ as identical and all points in $Y_0'$
as identical.

In terms of the auction with integer masses, this means that 
we only need one additional multi-bidder (with large mass)
to represent all projections of multi-objects to the diagonal,
and vice versa. Specifically, 
writing $X_0=\{x_1,\ldots,x_k\}$ for the off-diagonal multi-bidders, let $m_X$
denote their total mass. Let $Y_0=\{y_1,\ldots,y_\ell\}$ denote
the off-diagonal multi-objects
with total mass $m_Y$. We introduce one additional multi-bidder $Y_0':=\{x_{k+1}\}$
(representing all projections of multi-objects), with mass $m_Y$,
and one additional multi-object $X_0':=\{y_{\ell + 1}\}$ with mass $m_X$.
The bidder--object benefits are set up according to~\eqref{eq:c_q_tilde}
(recall that $x_i'$ denotes the projection of
$x_i$ onto the diagonal):

\[ b_{i,j} = \begin{cases}
    -\| x_i - y_j \|_{\infty}^q, & i \leq k \mbox{ and } j \leq \ell \\
    -\| x_i - x_i' \|_{\infty}^q, & i \leq k \mbox{ and } j = \ell + 1 \\
    -\| y_j - y_j' \|_{\infty}^q, & i = k+ 1 \mbox{ and } j \leq \ell \\
    0,  & i = k+1 \mbox{ and } j = \ell + 1 \\
\end{cases}\]

\paragraph{Implementation.}
We implemented a geometric version of the auction with integer masses,
where the best slices of the off-diagonal multi-objects 
are determined using one global k-d tree, similar to
Section~\ref{sec:wasserstein}. Here, each leaf of the k-d tree represents
a multi-object, and its weight corresponds to the price of its cheapest slice.
For a fixed off-diagonal multi-bidder,
we can compute an upper bound on the value of all
multi-objects stored in a subtree of the k-d tree.
During a search, we maintain a candidate set of slices
whose total mass exceeds the unassigned mass of the multi-bidder,
and we can prune a subtree
if that upper bound is below the value of the worst candidate.
The weights in the k-d tree are updated as in Section~\ref{sec:wasserstein}.
The additional information required to compute the price increases
are gathered by similar techniques; we omit the details.
We did not implement a non-geometric version using lazy heaps because
it would suffer from the same quadratic space complexity as in the standard auction.

Again we need to deal with the diagonal multi-object and multi-bidder separately. 
We maintain a heap with the slices of the diagonal multi-object sorted by the price and
a heap with the slices of all multi-objects (including the diagonal one) sorted by their value for the diagonal
bidder. The diagonal bidder finds the best slices by simply traversing the latter heap.
An off-diagonal bidder first uses the k-d tree to find the best slices of 
off-diagonal objects. Then it starts traversing the heap with slices
of the diagonal object, replacing the off-diagonal slices
with the diagonal ones as long as the diagonal slices offer better values.
When the value of the next diagonal slice in the heap is below the minimal value 
of the currently accumulated slices, we stop traversing the heap
with diagonal slices. When slice prices are increased, we immediately update
the heaps.

\paragraph{Experiments.}
As input, we turn the aforementioned instances of \dtype{normal}
type into diagrams with integer masses.
For each point of the original diagram, we assign mass $m$,
drawn uniformly from the range $[\lceil k/2\rceil, \lfloor 3k/2\rfloor]$,
so that the average mass of a point is $k$.
In our experiments, we compare the standard auction and the auction with integer masses
for $k=1,10,50,100$.

We generated \dtype{normal} instances with 1,000 to 10,000 points, in increments
of 1,000, with 10 instances per size.
Figure~\ref{fig:wasserstein:sop_weigts_1_10_50_100} shows the average running times.
There is an overhead for mass 1 (a factor of roughly 4.5 in the figure).
This ratio is not constant: the overhead becomes larger when the number of points grows.
We also observe that it depends on the parameters of the distribution from which the points
were drawn.
For average mass 10, the auction with masses is comparable to the standard auction.
For higher masses, 50 and 100, the advantage of the former is evident.

There is no clear dependence between the running time and the average mass.
We took 4 instances with 10000 points each and tried larger average masses (with the same
$[\lceil k/2\rceil, \lfloor 3k/2\rfloor]$ distribution).
Figure~\ref{fig:wasserstein_sop_weight_dep} illustrates the result.
We can see that the running time does not increase much when the average mass increases,
and may even decrease. That seems to depend very much on the particular instance 
and the distribution of masses inside it. 

The memory consumption of auction with integer masses usually scales
linearly with the number of points (for fixed average mass).
In principle, the memory size can grow proportional to the total
mass of the point sets when all slices shrink to size one,
but such intensive slicing did not appear in our examples.

\begin{figure*}[t]
    \centering
    \parbox{0.49\textwidth}{
    \begin{tikzpicture}
    \begin{axis}[
                    xtick={2000,4000,6000,8000,10000},
                    xmin=1000,xmax=10000,
                    height=2.5in,
                    width=.49\textwidth,
                    xlabel={\# multi-points},
                    ylabel={Seconds},
                    title={average mass = 1},
                    legend pos=north west,
                ]

        \addaveraged{weight_1_nonweighted.txt}
        \addlegendentry{standard}

        \addaveraged{weight_1_sop.txt}
        \addlegendentry{with masses}
    \end{axis}
    \end{tikzpicture}
    }
    \parbox{0.49\textwidth}{
    \begin{tikzpicture}
    \begin{axis}[
                    xtick={2000,4000,6000,8000,10000},
                    xmin=1000,xmax=10000,
                    height=2.5in,
                    width=.49\textwidth,
                    xlabel={\# multi-points},
                    ylabel={Seconds},
                    title={average mass = 10},
                    legend pos=north west,
                ]

        \addaveraged{weight_10_nonweighted.txt}
        \addlegendentry{standard}

        \addaveraged{weight_10_sop.txt}
        \addlegendentry{with masses}
    \end{axis}
    \end{tikzpicture}
    }\\

        \parbox{0.49\textwidth}{
    \begin{tikzpicture}
    \begin{axis}[
                    ymode=log,
                    xtick={2000,4000,6000,8000,10000},
                    xmin=1000,xmax=10000,
                    height=2.5in,
                    width=.49\textwidth,
                    xlabel={\# multi-points},
                    ylabel={Seconds},
                    title={average mass = 50},
                    legend pos=south east,
                ]

        \addaveraged{weight_50_nonweighted.txt}
        \addlegendentry{standard}

        \addaveraged{weight_50_sop.txt}
        \addlegendentry{with masses}
    \end{axis}
    \end{tikzpicture}
    }
    \parbox{0.49\textwidth}{
    \begin{tikzpicture}
    \begin{axis}[
                    ymode=log,
                    xtick={2000,4000,6000,8000,10000},
                    xmin=1000,xmax=10000,
                    height=2.5in,
                    width=.49\textwidth,
                    xlabel={\# multi-points},
                    ylabel={Seconds},
                    title={average mass = 100},
                    legend pos= south east,
                ]

        \addaveraged{weight_100_nonweighted.txt}
        \addlegendentry{standard}

        \addaveraged{weight_100_sop.txt}
        \addlegendentry{with masses}
    \end{axis}
    \end{tikzpicture}
    }

    \caption{Comparison of standard auction and auction with integer masses
             on \dtype{normal} data for mass $1$ (upper-left) and
             average masses $10$ (upper-right), $50$ (lower-left) and $100$ (lower-right).}
    \label{fig:wasserstein:sop_weigts_1_10_50_100}
\end{figure*}
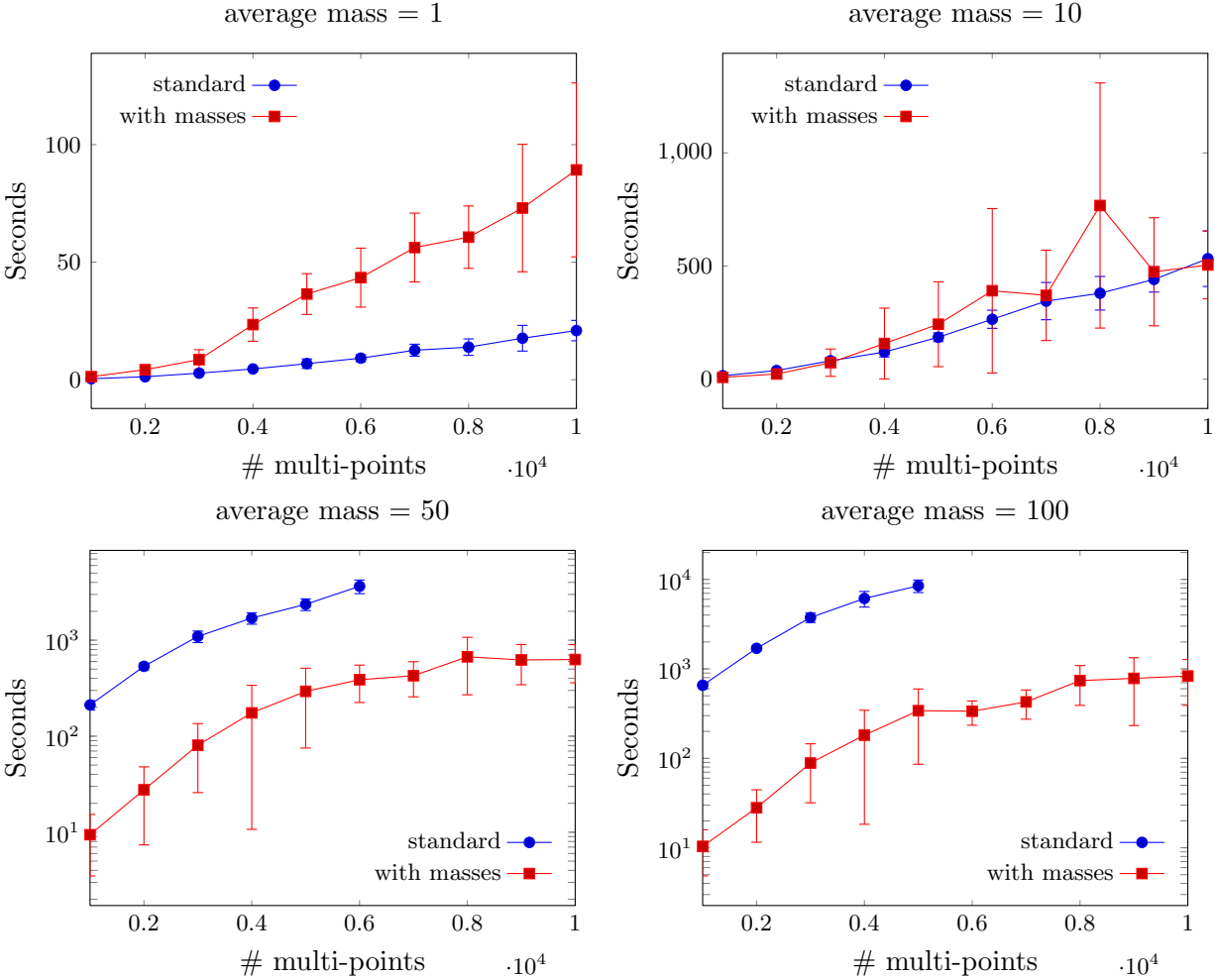

\begin{figure}[t]
    \centering
    \begin{tikzpicture}
    \begin{axis}[
                    xmode=log,
                    ymin=0,ymax=1500,
                    xtick={200,400,800,1600,3200,6400,12800},
                    xticklabels={200,400,800,1600,3200,6400,12800},
                    height=4.5in,
                    width=.79\textwidth,
                    xlabel={Average mass},
                    ylabel={Running time, s},
                    legend pos=north west,
                ]
        \addplot+[geommark] table[x index=0, y index=1] {large_weights_transpose.txt} ;
        \addlegendentry{Instance 1.}

        \addplot+[nongeommark] table[x index=0, y index=2] {large_weights_transpose.txt} ;
        \addlegendentry{Instance 2.}

        \addplot+[greentrianglemark] table[x index=0, y index=3] {large_weights_transpose.txt} ;
        \addlegendentry{Instance 3.}

        \addplot+[browndiamondmark] table[x index=0, y index=4] {large_weights_transpose.txt} ;
        \addlegendentry{Instance 4.}

    \end{axis}
    \end{tikzpicture}

    \caption{Dependence of the running time from the average mass for four particular instances of size 10000. Note the exponential scale on the $x$-axis.}
\label{fig:wasserstein_sop_weight_dep}
\end{figure}
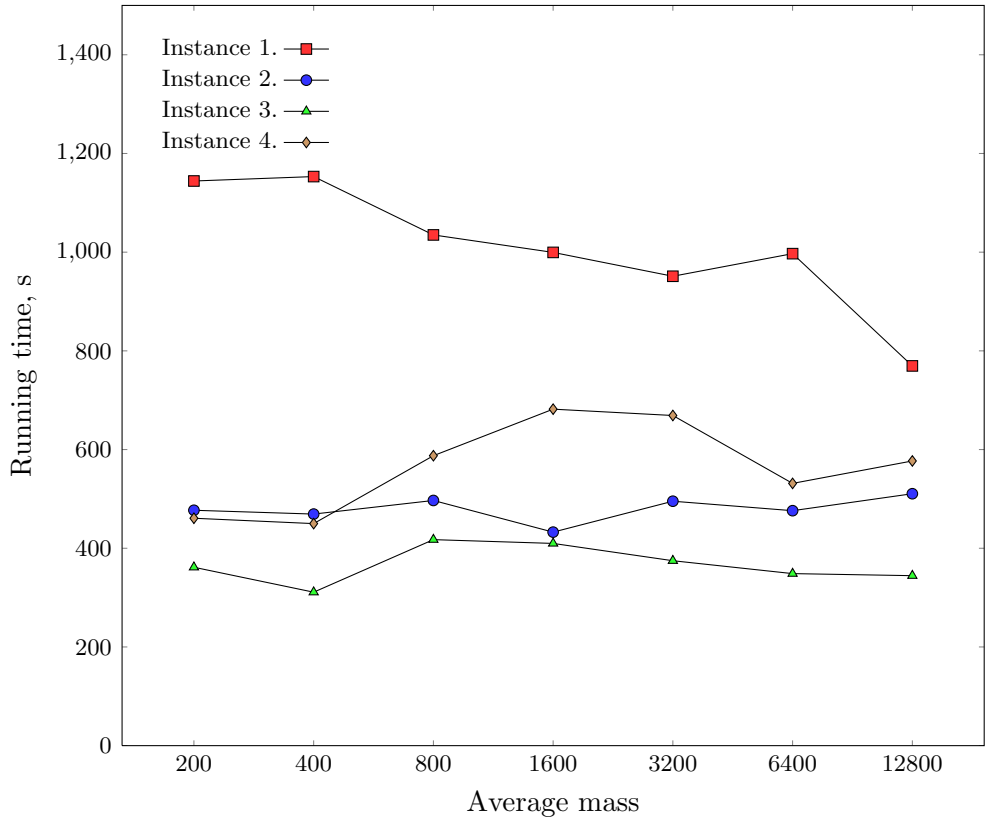

\section{Conclusion}
We have demonstrated that geometry helps to compute bottleneck and Wasserstein distances
of bipartite point sets in two dimensions. Our approach leads to a faster computation
of distances between persistence diagrams. Therefore, we expect our software
to have an immediate impact on the computational pipeline of topological data analysis.

For bottleneck matchings,
an interesting question would be how our k-d tree implementation compares
in practice with the (theoretically) more time efficient, but more space demanding alternative
of range trees, and with other point location data structures.

For Wasserstein matchings, we plan to further improve our implementation of the auction
algorithm, including a parallel version for large instances.
Simple heuristics can also improve special cases. For example, if $X$ and $Y$
are persistence diagrams and $S\subset X\cap Y$ is the set of common off-diagonal points,
it holds for $q=1$ that $W_q(X,Y)=W_q (X\setminus S,Y\setminus S)$, as one shows very easily.
This property allows to remove common points in the diagram before
applying the auction algorithm.
We also wonder how the auction approach compares with the various
alternatives proposed in ~\cite{bdm-assignment},
and for which of these approaches can geometry help compute the Wasserstein
distance efficiently, either exactly or approximately.

A natural approximation scheme for computing the Wasserstein distance for very large instances consists of
placing a finite grid over $\R^2$ and ``snapping'' points to their closest grid vertex. The result
is an instance with a potentially high multiplicity in each grid vertex.
The problem with this approach is the approximation error
introduced by the discretization step. A crude error bound is the total number of points in both diagrams
multiplied by the diameter of the grid cells. An interesting question is to evaluate more refined
discretization schemes with respect to their practical performance.

\section*{Acknowledgements}
We thank Sergio Cabello for pointing out that the worst-case complexity
of k-d trees and range trees remains valid under deletions of points,
and for further valuable remarks on an earlier draft of the paper.

Michael Kerber and Arnur Nigmetov acknowledge support by the Max Planck Center for Visual Computing and Communication.
Dmitriy Morozov is supported by Advanced Scientific Computing Research, Office of Science, U.S. Department of Energy,
under Contract DE-AC02-05CH11231.

\bibliographystyle{plain}
\bibliography{bib}

\end{document}